\documentclass[11pt,a4paper]{article}
\usepackage[margin=1in]{geometry}
\usepackage{tikz}
\usetikzlibrary{positioning,arrows.meta}
\usepackage{subcaption}

\usepackage[english]{babel}
\usepackage{algpseudocode}
\usepackage{makecell}
\usepackage[english]{babel}
\usepackage{array} 
\usepackage{booktabs} 
\usepackage{authblk} 

\usepackage[ruled,linesnumbered,vlined]{algorithm2e}
\usepackage{algpseudocode}
\usepackage{amssymb}
\usepackage{mathtools}

\usepackage{amsmath}
\usepackage{amsthm}
\usepackage{graphicx}
\usepackage[colorlinks=true, allcolors=blue]{hyperref}
\usepackage{natbib}
\usepackage{framed}
\definecolor{shadecolor}{gray}{0.9}
\usepackage{pgfplots}
\pgfplotsset{compat=1.17}
\usepackage{tikz}
 \usetikzlibrary{patterns}
\usepackage{pgfplotstable}
    \pgfplotsset{
    name nodes near coords/.style={
        every node near coord/.append style={
            name=#1-\coordindex,
            alias=#1-last,
        },
    },
    name nodes near coords/.default=coordnode
    }

\usepackage{multirow}

\usepackage[table]{xcolor} 

\newtheorem{theorem}{Theorem}[section]
\newtheorem{definition}[theorem]{Definition}

\newtheorem{corollary}[theorem]{Corollary}
\newtheorem{lemma}[theorem]{Lemma}

\newtheorem{example}[theorem]{Example}
\newtheorem{observation}[theorem]{Observation}

\newcommand{\bX}{\mathbf{X}}
\newcommand{\bB}{\mathbf{B}}

\newcommand{\bc}{\mathbf{c}}

\newcommand{\MMS}{\mathsf{MMS}}

\definecolor{ruleblue}{RGB}{70,120,200}
\definecolor{shadegray}{gray}{0.9}

\makeatletter
\renewcommand{\and}{\end{tabular}\hspace{3em}\begin{tabular}[t]{c}}
\makeatother
\begin{document}
\title{Allocating Chores with Restricted Additive Costs: \\
Achieving EFX, MMS, and Efficiency Simultaneously\thanks{The authors are ordered alphabetically. Xiaowei Wu is funded by the Science and Technology Development Fund (FDCT), Macau SAR (file no. 0147/2024/RIA2, 0014/2022/AFJ, 0085/2022/A, and 001/2024/SKL). 
This work is partially supported by the Ministry of Education, Singapore, under its Academic Research Fund Tier 1 (RG98/23).
}}

\author[1]{Zehan Lin}
\author[1]{Xiaowei Wu}
\author[2]{Shengwei Zhou}
\affil[1]{University of Macau, \texttt{\{yc47490,xiaoweiwu\}@um.edu.mo}}
\affil[2]{Nanyang Technological University, \texttt{shengwei.zhou@ntu.edu.sg}}
\date{}

\maketitle

\begin{abstract}
  In a web-based review platform, papers from various research fields must be assigned to a group of reviewers. 
  Each paper has an inherent cost, which represents the effort required for reading and evaluating it (e.g., the paper's length). 
  Reviewers can bid on papers they are interested in, and if they are assigned a paper they have bid on, no cost is incurred. Otherwise, the inherent cost $c(e)$ for paper $e$ applies. 
  We capture this with a model of restricted additive costs: every item $e$ has a cost $c(e)$, and each agent either incurs $0$ or $c(e)$ for $e$. 
  In this work, we study how to allocate such chores fairly and efficiently.
  We propose an algorithm for computing allocations that are both EFX and MMS.
  Furthermore, we show that our algorithm achieves a $2$-approximation of the optimal social cost, and the approximation ratio is optimal.
  We also show that slightly weaker fairness guarantees can be obtained if one requires the algorithm to run in polynomial time.
\end{abstract}

\maketitle

\section{Introduction}
Fair allocation problems focus on allocating a set $M$ of $m$ items among a group $N$ of $n$ agents, where agents have heterogeneous valuation functions over the items.
When the valuations are positive, the items are viewed as goods (e.g., resources), whereas negative valuations indicate that the items are chores (e.g., tasks).
The objective of fair allocation is to compute an allocation that ensures fairness while ideally also achieving efficiency.

In recent years, fair division has become increasingly relevant in web-based platforms, where resources or tasks need to be allocated among multiple agents with varying preferences. Examples include the fair allocation of advertising spaces on digital platforms, task assignments in collaborative online projects, and even the distribution of computational resources in cloud services. 
In such settings, ensuring fairness in the distribution of resources or responsibilities, while also optimizing social costs, presents a significant challenge.
The growing importance of such scenarios in online environments motivates the study of fair allocation algorithms.

In this work, we focus on fair allocations for chores, where each agent has an additive cost function that assigns a non-negative cost to each item.
Various notions have been proposed to measure the fairness of allocations.
Envy-freeness (EF) is one of the most extensively studied among them.
An allocation is EF if no agent prefers another agent’s bundle over her own.
However, EF allocations are not guaranteed to exist when items are indivisible.
As a result, much attention has been given to its relaxations.
Two widely studied relaxations are envy-freeness up to one item (EF1)~\cite{conf/sigecom/LiptonMMS04} and envy-freeness up to any item (EFX)~\cite{journals/teco/CaragiannisKMPS19}.
Specifically, EF1 requires that any envy between two agents can be eliminated by removing some item from the envious agent’s bundle, whereas EFX requires that envy can be eliminated by removing any item from the envious agent’s bundle.

\paragraph{EFX Allocations for Chores.}
The existence of EFX allocations is regarded as the central problem in the field of fair allocation and referred to as “fair allocation’s most enigmatic question” by Procaccia~\cite{journals/cacm/Procaccia20}.
For two agents with general cost functions, EFX allocations can be computed by the divide-and-choose algorithm. 
However, the existence of EFX allocations remains unknown for $n\geq 3$ agents with additive functions.
A sequence of works tried to answer this question by considering important special cases.
Aziz et al.~\cite{conf/atal/0001LRS23} demonstrated that EFX allocations can be computed when the set of chores can be divided into two types, such that all agents perceive any two items of the same type as identical.
For identical ordinal preference (IDO) instances, where all agents share the same preference ordering over items, Aziz et al.~\cite{journals/ai/AzizLMWZ24} showed that EFX allocations are guaranteed to exist.
Tao et al.~\cite{journals/tcs/TaoWYZ25} proved the existence of EFX allocations for binary instances, where the cost of any item to any agent is either $0$ or $1$.
The existence of EFX allocations has also been established under other structured settings, such as leveled preferences~\cite{journals/teco/GafniHLT23}, lexicographic preferences~\cite{conf/atal/HosseiniSVX23}, and when $m \leq 2n$~\cite{conf/stoc/GargMQ25, journals/tcs/KobayashiMS25}.
Nevertheless, extending these results to more general settings remains highly non-trivial. 
Even for the seemingly modest case of bi-valued instances (which generalizes the binary instances by allowing item costs to be either $1$ or some fixed constant $\epsilon \in (0, 1)$), existence results have been established only in the case of three agents or when $\epsilon = 0.5$, despite significant efforts~\cite{journals/ai/ZhouW24,journals/tcs/KobayashiMS25,conf/ijcai/GargMQ23, journals/corr/abs-2501-04550}.



\paragraph{MMS Allocations for Chores.} 
Another well-studied fairness notion is the maximin share (MMS), introduced by Budish~\cite{conf/bqgt/Budish10}.
For each agent, the MMS is defined as the minimum cost she can guarantee for herself by partitioning all items into 
$n$ bundles and then receiving the least preferred bundle.
When $n = 2$, the divide-and-choose algorithm can be used to compute MMS allocations. 
However, for $n \geq 3$, MMS allocations are not guaranteed to exist~\cite{conf/wine/FeigeST21, conf/aaai/AzizRSW17}.
Consequently, much of the research has shifted toward computing approximate MMS allocations, e.g., see~\cite{conf/aaai/AzizRSW17,journals/mp/AzizLW24,conf/sigecom/FeigeH23,journals/teco/BarmanK20,conf/sigecom/HuangL21}.
The state-of-the-art approximation ratio of MMS for chores is $13/11$ by Huang and Segal-Halev~\cite{conf/sigecom/HuangS23}.
On the negative side, Feige et al.~\cite{conf/wine/FeigeST21} showed that no algorithm can achieve an approximation ratio better than $44/43$, even for three agents.

\paragraph{Restricted Additive Functions.} 
In the restricted additive setting\footnote{For simplicity, in the remainder of this paper we refer to the restricted additive setting as the restricted setting, when the context is clear.}, each item $e$ has an inherent value $c(e)$ so that for any agent $i$, we have $c_i(e) \in \{0, c(e)\}$.
The restricted setting was extensively studied in prior work, which mainly focus on the case of goods.
The setting has been well-studied in the Santa Claus problem~\cite{journals/talg/AsadpourFS12, journals/rsa/SahaS18, conf/stoc/BateniCG09, conf/focs/ChakrabartyCK09}, which focuses on the computational complexity of maximizing the minimum utility. 
The setting in fair allocation has also drawn significant attention recently, e.g.,
Akrami et al.~\cite{conf/ijcai/AkramiRS22} showed that EF$2$X\footnote{EF$2$X allocations for goods require that the envy between two agents can be eliminated by removing any two items from the envied agent.} allocations for goods exist when agents have restricted additive functions.
The result was partially improved by Kaviani et al.~\cite{journals/corr/abs-2407-05139}, who showed that EFX allocations can be computed when each item is only relevant to two agents.
Mashbat et al.~\cite{conf/sagt/BotanRSW23} proved the existence of MMS allocations for three agents under restricted setting.
However, very few works have considered the setting for chores. 
The work by Sun and Chen~\cite{journals/eor/SunC25} is a notable exception, who proposed a group-strategyproof randomized mechanism with best-of-both-worlds fairness guarantee.
Therefore, it remains unknown whether EFX allocations exist for chores in the restricted setting.
Restricted additive functions for chores can arise in many scenarios. 
For instance, in a web-based review platform, papers must be assigned to a group of reviewers. 
Each paper has an inherent cost, which represents the effort required to evaluate it (e.g., the paper’s length). 
Reviewers can bid on papers they are interested in, and if they are assigned a paper they bid on, no cost is incurred. Otherwise, the inherent cost $c(e)$ for paper $e$ applies.


\paragraph{Efficiency.}
Besides fairness, efficiency is another important measure of the quality of allocations. 
Pareto optimality (PO) is one of the most widely used measures of efficiency. 
An allocation is PO if no other allocation can make at least one agent strictly better off without making any other agent worse off. 
Another important efficiency measure is the social cost, which is the total (sum of) cost incurred by all agents.
The optimal social cost can be achieved by allocating each item to the agent who incurs the least cost for it.
However, this often results in a highly unfair allocation.
In recent years, the existence of allocations that are both fair and efficient has drawn significant attention. 
Unfortunately, efficiency and fairness often conflict with each other.
Requiring an allocation to satisfy certain fairness guarantees may lead to a loss in efficiency, i.e., the resulting social cost may be higher than the (unconstrained) optimal.
To quantify the increase in social cost due to the fairness constraint, Bertsimas et al.~\cite{journals/ior/BertsimasFT11} and Caragiannis et al.~\cite{journals/mst/CaragiannisKKK12} introduced the price of fairness (PoF), which is the ratio between the optimal social cost of fair allocations and the unconstrained optimum.

\subsection{Discussion on Existing Algorithms}

We first discuss some natural algorithms for computing EFX allocations for chores in other related settings, and show some difficulties of extending these algorithms to the restricted setting.
Via these discussions, we identify several useful observations that will inspire and guide the design of our algorithm.
Observe that when $c(e) = 1$ for all $e\in M$, the restricted setting degenerates to the binary setting (in which $c_i(e)\in \{0,1\}$ for all $i\in N$ and $e\in M$). 
Therefore, it is natural to consider extensions of algorithms for computing EFX allocations for binary instances to the restricted setting.
In the algorithm by Tao et al.~\cite{journals/tcs/TaoWYZ25}, the items with cost $0$ to some agent are allocated first, which is natural and reasonable, as they do not incur any cost to the receiving agent.
We use $M^0$ to denote this set of items.
Then they try to allocate items with non-zero costs to all agents one by one.
During this process, previously assigned items may be reallocated.
Their algorithm guarantees that every item $e\in M^0$ will be allocated to an agent with zero cost on $e$ after the reallocations.
Therefore, in addition to guaranteeing EFX, the returned allocation is PO as it minimizes social cost. However, they also give a simple example showing that it is not possible to extend their algorithm to the more general restricted setting:

\begin{table}[h!] 
\centering 
\begin{tabular}{lccc} 
\toprule
\textbf{Agent} & \textbf{$e_1$} & \textbf{$e_2$} & \textbf{$e_3$} \\
\midrule
\ \ \ \ 1  & $1 + \epsilon$ & $1$ & $0$ \\
\ \ \ \ 2  & $1 + \epsilon$ & $0$ & $1$ \\
\bottomrule
\end{tabular}
\caption{A hard instance modified from~\cite{journals/tcs/TaoWYZ25}, where $\epsilon > 0$.} \label{Table: NonEFXandPO}
\end{table}

In the above instance, if we allocate item $e_2$ to agent $2$ and $e_3$ to agent $1$, then no matter who receives item $e_1$, the resulting allocation is not EFX.
In other words, to ensure EFX, we need to allocate either $e_2$ to agent $1$ or $e_3$ to agent $2$.
Indeed, allocating $e_1$ to agent $1$ and the remaining items to agent $2$ gives an EFX allocation.

\begin{observation} \label{observation:allocate-non-zero}
    To compute EFX allocations, we sometimes need to allocate an item $e$ to agent $i$ with $c_i(e) > 0$, even if there exists another agent $j\neq i$ with $c_j(e) = 0$.
\end{observation}

Moreover, it can be verified that every PO allocation must allocate item $e_2$ to agent $2$ and item $e_3$ to agent $1$, which implies that EFX is incompatible with PO in the restricted setting.

\begin{observation} \label{observation:EFX-PO-incompatible}
    EFX and PO allocations may not exist in the restricted setting.
\end{observation}

We also consider other natural algorithms, e.g., the algorithm of Aziz et al.~\cite{journals/ai/AzizLMWZ24} for computing EFX allocations for IDO instances based on envy-cycle elimination.
Their algorithm allocates items one by one based on a directed envy graph, where each vertex represents an agent, and a directed edge from agent $i$ to agent $j$ indicates that $i$ envies $j$.
In each round, an agent who currently does not envy any other agent is chosen to receive the unallocated item with maximum cost, where such an agent is guaranteed to exist by eliminating directed cycles in the graph.
Their algorithm guarantees EFX since every agent is EF towards all other agents after removing the last received item, and the last received item is the item with the smallest cost by the property of IDO instances.
Naturally, we can extend the algorithm to the restricted setting, where the ordering of items is induced by $c$, e.g., we allocate items in decreasing order of $c(e)$.
Unfortunately, the algorithm fails to compute EFX allocations, because if an agent $i$ received an item $e$ with cost $0$, then even if agent $i$ is EF towards all other agents, we cannot pick $i$ to receive an item without violating EFX.
Consider the following simple example:

\begin{table}[h!] 
\centering 
\begin{tabular}{lcccc} 
\toprule
\textbf{Agent} & \textbf{$e_1$} & \textbf{$e_2$} & \textbf{$e_3$} & \textbf{$e_4$}\\
\midrule
\ \ \ \ 1  & $1 + \epsilon$ & $1$ & $0$ & $1$ \\
\ \ \ \ 2  & $1 + \epsilon$ & $0$ & $1$ & $1$ \\
\bottomrule
\end{tabular}
\caption{A hard instance for the envy-cycle elimination algorithm by Aziz et al.~\cite{journals/ai/AzizLMWZ24}.} \label{Table:not-EFX-for-ECE}
\end{table}

Following the algorithm, item $e_1$ will be allocated to one agent (say, agent $1$) while $e_2$ and $e_3$ will be allocated to the other agent (say, agent $2$).
Then agent $2$ will be chosen again to receive item $e_4$, which results in an allocation that is not EFX, since agent $2$ has received an item with cost $0$.
Note that allocating item $e_4$ to agent $1$ will also violate EFX.
From this example, we observe that we should be very careful in deciding the agents to receive items with cost $0$, as such an operation will strengthen the fairness requirement for the agent from EFX to EF.

\begin{observation} \label{observation:EFX-becomes-EF}
    If an agent $i$ receives item $e$ with $c_i(e) = 0$, then we need to ensure that agent $i$ is envy-free towards others.
\end{observation}

In summary, given the difficulty of extending existing work to the restricted setting, the existence of EFX allocations for chores remains open.
Moreover, since EFX and PO are incompatible, it is unknown whether we can have efficiency guarantees in addition to being EFX.

\subsection{Our Results}

In this work, we answer the above questions by proposing an algorithm for the computation of EFX allocations.
In addition, we show that the returned allocation is also MMS and achieves an optimal approximation ratio for the optimal social cost.

\medskip
\noindent
{\bf Main Result} (Theorem~\ref{theorem: EFX} and Theorem~\ref{theorem: effiency}){\bf .} \label{Result1}
{\em There exists an algorithm for computing allocations that are \emph{EFX}, \emph{MMS} and achieve a $2$-approximation of the optimal social cost for any restricted instance. Furthermore, the approximation ratio is optimal.}

\medskip

Compared to EFX, MMS is relatively easy to obtain under the restricted setting.
Therefore, in the following, we mainly discuss the algorithmic ideas for the computation of EFX allocations.

Motivated by Observation~\ref{observation:allocate-non-zero}, our algorithm takes a different approach from that of Tao et al.~\cite{journals/tcs/TaoWYZ25} and begins by allocating items $e \in M \setminus M^0$, which we refer to as $M^+$.
Note that for all $i \in N$ and $e \in M^+$, we have $c_i(e) = c(e)$.
Therefore, we can compute an EFX partition $\bB = \{B_1,B_2,\ldots,B_n\}$ for $M^+$, and it remains to allocate items in $M^0$.

A key observation is that we should not decide the assignment of bundles in $\bB$ \emph{before} the allocation of $M^0$: by Observation~\ref{observation:EFX-becomes-EF}, an agent cannot receive further zero-cost items unless she receives the best bundle in $\bB$.
Motivated by the matching-based approach~\cite{journals/ior/AkramiACGMM25,journals/talg/GargKK23,journals/corr/abs-2404-18133,journals/tcs/KobayashiMS25} for computing fair allocations, we allocate items in $M^0$ to bundles in $\bB$ without deciding the assignment of bundles.
Specifically, we let agents take turns to add unallocated items they consider having zero cost to their \emph{favorite} bundle in $\bB$.
We show that the EFX-representing graph between the agents and the final bundles (where there is an edge between an agent and a bundle if the bundle is EFX-feasible to the agent) admits a perfect matching, which gives an EFX allocation.
Since we may inevitably have to allocate items in $M^0$ to an agent with non-zero cost on the item (see Observation~\ref{observation:allocate-non-zero} and~\ref{observation:EFX-PO-incompatible}), to further guarantee a small social cost, we need to bound the increase in social cost due to these items.
Fortunately, by the design of our algorithm, we can show that agents receiving these items are EF towards all other agents, and there are at most $n/2$ such agents.
This allows us to establish a $2$ approximation of the optimal social cost, which is optimal: for the hard instance shown in Table~\ref{Table: NonEFXandPO}, the optimal social cost is $1+\epsilon$ while every EFX allocation has a social cost of $2+\epsilon$.
Our result implies the price of EFX is $2$ in the restricted settings. 
Interestingly, this stands in sharp contrast to prior work, where the price of EFX (and even EF1) has been shown to be unbounded in general instances~\cite{journals/aamas/SunCD23a}.

However, our algorithm is not guaranteed to terminate in polynomial time, e.g., it is NP-hard to compute MMS allocations even for identical agents.
Therefore, we also propose a polynomial-time algorithm that computes allocations that are EFX and $4/3$-MMS (see Section~\ref{ssec:poly-time-alg}).

\subsection{Other Related Works}

Due to the vast literature on fair allocation, in the following, we only review some of the most related works. For a more comprehensive review, please refer to the recent surveys by Aziz et al.~\cite{journals/sigecom/AzizLMW22} and Amanatidis et al.~\cite{journals/ai/AmanatidisABFLMVW23}.

\paragraph{Approximate EFX Allocations.} 
Given the difficulty of computing exact EFX allocations, researchers have focused on computing approximate EFX allocations.
Zhou and Wu~\cite{journals/ai/ZhouW24} presented a polynomial-time algorithm that computes a $(2 + \sqrt{6})$-approximate EFX (i.e., $(2 + \sqrt{6})$-EFX) allocation for three agents.
The ratio has been improved to $2$ by Afshinmehr et al.~\cite{journals/corr/abs-2410-18655} and Christoforidis et al.~\cite{conf/ijcai/ChristoforidisS24}.
For general additive instances, Garg et al.~\cite{conf/stoc/GargMQ25} proved the existence of $4$-EFX allocations.
This result was further improved to $2$-EFX by Garg et al.~\cite{DBLP:journals/corr/2EFX}.
For bi-valued instances (where the cost of an item to an agent is either $1$ or $\epsilon$), they showed that $3$-EFX and PO allocations always exist.
The ratio has recently been improved to $2-\epsilon$ by Lin et al.~\cite{journals/corr/abs-2501-04550}.

\paragraph{EFX Allocations for Goods.} 
Compared to the allocations for chores, the case of goods admits more fruitful results.
The EFX allocations are shown to exist for two agents with arbitrary valuations or any number of agents with identical valuations by Plaut et al.~\cite{journals/siamdm/PlautR20}; three agents by Chaudhury et al.~\cite{journals/jacm/ChaudhuryGM24} and Akrami et al.~\cite{journals/ior/AkramiACGMM25}; two types of agents by Mahara~\cite{journals/dam/Mahara23}; allocations on graphs by Christodoulou et al.~\cite{conf/sigecom/0001FKS23}; two types of items by Gorantla et al.~\cite{conf/soda/GorantlaMV23} and lexicographic preferences by Hosseini et al.~\cite{conf/aaai/HosseiniSVX21}.
Recently, Prakash et al.~\cite{journals/corr/abs-2410-13580} showed that EFX allocations exist for three types of agents.
For bi-valued instances, Amanatidis et al.~\cite{journals/tcs/AmanatidisBFHV21} and Garg et al.~\cite{journals/tcs/GargM23} showed that EFX and PO allocations can be computed efficiently.
Regarding the approximation of EFX allocations, Chan et al.~\cite{conf/ijcai/Chan00W19} demonstrated the existence of $0.5$-EFX allocations, while Amanatidis et al.~\cite{journals/tcs/AmanatidisMN20} improved the ratio to $0.618$. 
Recently, Amanatidis et al.~\cite{conf/sigecom/AmanatidisFS24} improved this ratio further to $2/3$ for some special cases.

\section{Preliminary} \label{section: preliminary}
We study the problem of fairly allocating a set of $m$ indivisible items (chores), denoted by $M$, to a group of $n$ agents $N$.
We refer to a subset of items $X \subseteq M$ as a bundle.
Each agent $i \in N$ is associated with an additive cost function $c_i: 2^M \to \mathbb{R}^+ \cup \{0\}$, which assigns a cost to every bundle of items.
We use $\bc = (c_1, \dots, c_n)$ to denote the cost functions of all agents.
For any subset $X \subseteq M$ and item $e \in M$, we use $X + e$ to denote $X \cup \{e\}$, and $X - e$ to denote $X \setminus \{e\}$.
For any $M' \subseteq M$, a partition $\bB = (B_1, \dots, B_n)$ is a collection of disjoint subsets of $M'$ such that $B_i \cap B_j = \emptyset$ for all $i \neq j$ and $\bigcup_{i \in N} B_i = M'$. 
An allocation $\bX = (X_1, \dots, X_n)$ is an ordered partition of the item set $M$ into $n$ disjoint bundles, where agent $i$ receives bundle $X_i$. 
In this paper, we distinguish the concepts of partition and allocation. 
A partition consists of $n$ unordered bundles whose assignments (to the agents) have not been decided.
On the other hand, in an allocation, the bundles are ordered and indexed by the agents receiving the bundles.
For any integer $t \geq 1$, we use $[t]$ to denote the set $\{1, \dots, t\}$. 
Throughout the paper, we use $i \in N$ to denote the agents and $i \in [n]$ to index bundles in a partition.
Given an instance $\mathcal{I} = (N, M, \bc)$, our goal is to compute an allocation $\bX$ that is both fair and efficient.

In this work, we focus on instances where each agent has a restricted cost function, which is formally defined as follows.

\begin{definition}[Restricted Functions]
    A set $\{c_1, c_2, \dots, c_n\}$ of cost functions is restricted, if there exists a function $c$ such that for any $i \in N$, any item $e \in M$, we have $c_i(e) \in \{0, c(e)\}$.
\end{definition}

We refer to any instance in which all agents have such cost functions as a restricted instance\footnote{Another natural way to extend the notion of restricted setting to chores is to assume that for every $i \in N$, $e \in M$, we have $c_i(e) \in \{\infty, c(e)\}$, where $\infty$ indicates that the agent cannot receive the item. However, we show that EFX allocations are not guaranteed to exist under this setting in the Appendix~\ref{section: infty}.}.
%
Note that when $c(e) = 1$ for all $e \in M$, restricted instances reduce to binary instances.
Next, we define the fairness notions. 

\begin{definition}[EF]
An allocation $\bX$ is envy-free (EF) if for any agents $i, j \in N$, we have $c_i(X_i) \leq c_i(X_j)$.
\end{definition}

\begin{definition}[EF1]
An allocation $\bX$ is envy-free up to one item (EF1) if for any agents $i, j \in N$, either $X_i = \emptyset$, or there exists an item $e \in X_i$ such that $c_i(X_i - e) \leq c_i(X_j)$.
\end{definition}

\begin{definition}[EFX]
An allocation $\bX$ is envy-free up to any item (EFX) if for any agents $i, j \in N$, either $X_i = \emptyset$, or for any item $e \in X_i$, we have $c_i(X_i - e) \leq c_i(X_j)$.
\end{definition}

\begin{definition}[MMS]
    Given a set of items $M' \subseteq M$, for any $i \in N$, her maximin share (MMS) value of items $M'$ is defined as
    $$
    \MMS_i (M') :=\min_{\bX \in \prod_n(M')} \max_{j \in N}  \ \left\{c_i(X_j)\right\},
    $$
    where $\prod_n(M')$ denotes the set of all $n$-partitions of $M'$.
\end{definition}

For convenience, we use
$\text{MMS}_i$ to denote $\text{MMS}_i(M)$ when $M$ is clear from the context.

\begin{definition}[$\alpha$-MMS]
    An allocation $\bX$ is said to satisfy the $\alpha$-approximate maximin share guarantee for some $\alpha \geq 1$ if, for every agent $i \in N$,
    $$
        c_i(X_i) \leq \alpha \cdot \MMS_i.
    $$
    Specifically, when $\alpha = 1$, the allocation is an \emph{MMS} allocation.
\end{definition}

Note that an $n$-partition of $M$ is called an $\alpha$-MMS partition for agent $i$ if every bundle in the partition has cost at most $\alpha \cdot \MMS_i(M)$.
%
%
For the convenience of describing whether an agent is satisfied with a bundle, we define the concept of feasibility.

\begin{definition}[$\alpha$-MMS-feasible] 
We say that a bundle $B\subseteq M$ is $\alpha$-MMS-feasible to agent $i$ if and only if $c_i(B) \leq \alpha \cdot \MMS_i$. 
\end{definition}

\begin{definition}[EFX-feasible] 
Given a partition $\textbf{B} = (B_1, \dots, B_n)$, a bundle $B_i$ is EFX-feasible to agent $j$ if and only if
\begin{equation*}
    \max_{e \in B_i}\{c_j(B_i - e)\} \leq \min_{k \in [n]} \{c_j(B_k)\}.
\end{equation*}
\end{definition}
Note that whether a bundle $B$ is $\alpha$-MMS-feasible to agent $i$ depends only on agent $i$ and bundle $B$; however, whether $B$ is EFX-feasible also depends on the partition that contains $B$.
For convenience, for agent $i$ and bundle $B$, we further define
\begin{equation*}
    \hat{c}_i(B) = \max_{e \in B}\{c_i(B - e)\}
\end{equation*}
as the bundles cost after removing the item with the smallest cost.

\smallskip

Next, we introduce the efficiency notions.

\begin{definition}[Social Cost]
     The social cost of an allocation $\textbf{X}$ is defined as 
     $sc(\textbf{X}) = \sum_{i \in N}c_i(X_i)$.
\end{definition}

\begin{definition}[Pareto Optimal]
    An allocation $\bX'$ \emph{Pareto dominates} another allocation $\bX$ if $c_i(X_i') \leq c_i(X_i)$ for all $i \in N$ and the inequality is strict for at least one agent. 
    An allocation X is \emph{Pareto optimal} (PO) if $\bX$ is not dominated by any other allocation.
\end{definition}

\section{EFX and MMS Allocations}

In this section, we investigate how to compute allocations that are both EFX and MMS for restricted instances.
The main result of this section is summarized as follows.

\begin{theorem} \label{theorem: EFX}
    For restricted instances of indivisible chores, there exists an algorithm that computes allocations that satisfy both \emph{EFX} and \emph{MMS} to all agents.
\end{theorem}

The remainder of this section is devoted to proving the above theorem.
To introduce our algorithm, we first classify the set of chores into two categories, as defined below.

\begin{definition}
    Chores are categorized as:
    \begin{itemize}
        \item $M^+ = \{e \in M: c(e) > 0 \text{ and } \forall i \in N, c_i(e) = c(e)\}$, which represents the set of items that have a consistent cost across all agents.
        \item $M^0 = M \setminus M^+$ denotes the set of remaining items, where each item incurs zero cost for at least one agent.
        Furthermore, for each agent $i \in N$, we use $M_i^0 = \{e \in M : c_i(e) = 0\}$ to denote the set of items that incur zero cost to agent $i$.
        We have $M^0 = \bigcup_{i\in N} M_i^0$.
    \end{itemize}
\end{definition}

For convenience, we refer to items in $M^0_i$ as \emph{zero-cost} items for agent $i$. 
We note that, by definition, items in $M^+$ induce an instance where all agents have identical cost functions. 
For such instances, EFX allocations can be computed efficiently.
Additionally, there may be some overlap between $M_i^0$ and $M_j^0$ for two agents $i$ and $j$, meaning that an item could be a zero-cost item for multiple agents.

\subsection{Description of the Algorithm}

Our algorithm has three main phases (see Algorithm~\ref{alg:restricted} for details).

\paragraph{Phase 1: Allocate items in $M^+$.}
We first compute an MMS partition $\bB^+ = (B_1^+,\ldots,B_n^+)$ with respect to (w.r.t.) $c$ over items in $M^+$. 
Then, as long as there exists any bundle that is not EFX-feasible w.r.t. $c$, we perform a reallocation to balance the bundles without violating the MMS-feasibility.
We repeat the procedure until every bundle is EFX-feasible. 
We show in Lemma~\ref{lemma: Phase1M+EFX} that Phase 1 terminates after a finite number of steps and returns a partition in which every bundle $B' \in \bB^+$ is both EFX-feasible and MMS-feasible to any agent.

\paragraph{Phase 2: Allocate items in $M^0$.}
In this phase, we allocate items in $M^0$, and turn the partition $\bB^+$ into $\bB = (B_1,\ldots,B_n)$.
During this phase, agents take turns to add the unallocated zero-cost items to their most preferred bundle. 
Specifically, in agent $i$'s turn, the items in $M_i^0$ that have not been allocated are added to the bundle $B_j$ that has the minimum $c_i(B_j)$.
We say that agent $i$ \emph{modifies} bundle $B_j$.
Note that it is possible that all items in $M_i^0$ have already been allocated.
However, in this case, we still consider it a modification.
Throughout Phase $2$, we maintain a vector $L = (L_1, \dots, L_n)$ to keep track of the \emph{last} agent that has modified a bundle.
Specifically, we initialize $L = (0, \dots, 0)$, and update $L_j = i$ when agent $i$ modifies $B_j$.
Therefore, $L_j = i\in N$ indicates that agent $i$ is the last agent that has modified bundle $B_j$;
$L_j = 0$ indicates that bundle $B_j$ is not modified by any agent during Phase $2$, and we have $B_j = B_j^+$.

At the end of Phase $2$, we show the following: 
\begin{itemize}
    \item For any bundle $B_j \in \bB$ such that $L_j \neq 0$, it is MMS-feasible and EFX-feasible to agent $L_j$ (see Lemma~\ref{lemma: M0EFX}). 
    \item For any bundle $B_j \in \bB$ with $L_j = 0$, which remains unchanged during this phase, it is MMS-feasible and EFX-feasible to all agents (see Lemma~\ref{lemma: M+EFX}).
\end{itemize}

\paragraph{Phase 3: Compute the final allocation.}
Finally, in this phase, we decide the assignment of bundles and turn the partition $\bB$ into an allocation $\bX = (X_1,\ldots,X_n)$.
For any bundle $B_j \in \bB$ with $L_j \neq 0$, we assign $B_j$ to agent $L_j$. 
As we will show below, $B_j$ is the most preferred bundle for agent $L_j$ and therefore agent $L_j$ is EF towards all other agents.
For the remaining bundles (with $L_j = 0$), we allocate them to the remaining agents (who have not received any bundle) arbitrarily.
Since these bundles only contain items from $M^+$, we use Lemma~\ref{lemma: M+EFX} to show that the bundles are MMS-feasible and EFX-feasible to their receivers.

\medskip

We illustrate the execution of our algorithm in Figure~\ref{Fig: Algorithm}, where we show the modifications of agents to the bundles at the end of Phase $2$.
We use a blue edge between an agent and a bundle to represent a modification, where only the last agent's edge is solid. 

\begin{figure}[htb]
\centering
\begin{tikzpicture}[scale=0.5, transform shape, >=Stealth, node distance=2.5cm and 1.2cm]
\node (b1) [rectangle, draw, minimum width=1.2cm, minimum height=1.2cm, font=\large] {B\textsubscript{1}};
\node (b2) [rectangle, draw = blue, thick, right=of b1, minimum width=1.2cm, minimum height=1.2cm, font=\large] {B\textsubscript{2}};
\node (b3) [rectangle, draw, right=of b2, minimum width=1.2cm, minimum height=1.2cm, font=\large] {B\textsubscript{3}};
\node (b4) [rectangle, draw = blue, thick, right=of b3, minimum width=1.2cm, minimum height=1.2cm, font=\large] {B\textsubscript{4}};
\node (b5) [rectangle, draw, right=of b4, minimum width=1.2cm, minimum height=1.2cm, font=\large] {B\textsubscript{5}};
\node (b6) [rectangle, draw=blue, thick, right=of b5, minimum width=1.2cm, minimum height=1.2cm, font=\large] {B\textsubscript{6}};

\node (i1) [circle, draw, below=of b1, minimum size=1.2cm, font=\large] {1};
\node (i2) [circle, draw, below=of b2, minimum size=1.2cm, font=\large] {2};
\node (i3) [circle, draw, below=of b3, minimum size=1.2cm, font=\large] {3};
\node (i4) [circle, draw, below=of b4, minimum size=1.2cm, font=\large] {4};
\node (i5) [circle, draw, below=of b5, minimum size=1.2cm, font=\large] {5};
\node (i6) [circle, draw, below=of b6, minimum size=1.2cm, font=\large] {6};

\draw[dashed, thick, blue] (i5) -- (b6);
\draw[dashed, thick, blue] (i1) -- (b2);
\draw[dashed, thick, blue] (i2) -- (b2);

\draw[-, thick, blue] (i4) -- (b4);
\draw[-, thick, blue] (i3) -- (b2);
\draw[-, thick, blue] (i6) -- (b6);
\end{tikzpicture}
\caption{Illustration for the modifications by the end of Phase 2. \label{Fig: Algorithm}
}
\end{figure}
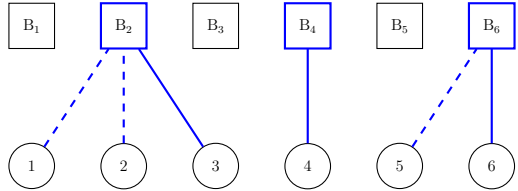

\begin{algorithm}[!htb]
\caption{Computation of EFX and MMS Allocations for Restricted Instances}
\label{alg:restricted}
\KwIn{Restricted instance $I = (N, M, c, \{c_i\}_{i \in N}$)}
    Initialize $B_i \gets \emptyset$ for all $i \in [n]$; \\
    \tcp{Phase 1: Allocate items in $M^+$} 
    Compute an MMS partition $\textbf{B} = (B_1, \dots, B_n)$ of $M^+$ w.r.t. $c$; \\
    \tcp{Modify the partition to ensure EFX\footnotemark};
    \While{$\max_{B \in \bB} \left\{\hat{c}(B)\right\} > \min_{B \in \mathbf{B}} \{c(B)\}$ \label{While-loop: EFX}}{ 

    $B_b \gets \arg\max_{B \in \bB} \left\{\hat{c}(B)\right\}$; \\
    $B_l \gets \arg\min_{B \in \mathbf{B}} \{c(B)\}$; \\
    $e \gets \arg \min_{e' \in B_b} \{c(e')\}$;\\
    $B_l \gets B_l + e$;\\    
    $B_b \gets B_b - e$;\\
    }
    \tcp{Phase 2: Allocate items in $M^0$} 
    $R \leftarrow M^0$; \hfill \tcp{Remaining items}
    For all $i \in [n]$, initialize $L_i \gets 0$\; 
    \For{$i \in N$}{
        $S_i \leftarrow M_i^0 \cap R$; \\
        $B_j \leftarrow \arg \min_{B \in \bB} \{c_i(B)\}$; \hfill \tcp{Break ties arbitrary} 
        $B_j \leftarrow B_j \cup S_i$; \\
        $R \leftarrow R \setminus S_i$;\\
        $L_j \leftarrow i$; \\
    }
    \tcp{Phase 3: Compute the final allocation}
    \For{$i \in [n]$}{ \label{Initial allocation}
        \If{$L_i \neq 0$}{
            $X_{L_i} \leftarrow B_i$; \\
            $N \leftarrow N \setminus \{L_i\}$; \\
            $\bB \leftarrow \bB \setminus \{ B_i \}$; \\
        }
    }
    \For{$i \in N$}{ \label{final allocation} 
        Pick arbitrary bundle $B \in \bB$; \\
        $X_i \leftarrow B$; \\
        $\bB \leftarrow \bB \setminus \{ B\}$; \\
    }
\KwOut{EFX and MMS allocation $\bX$}
\end{algorithm}
\footnotetext{Recall that $\hat{c}_i(B) = \max_{e \in B}\{c_i(B - e)\}$}
\subsection{Fairness Guarantee} \label{Section: Fairness}

Now we prove the fairness guarantee of the returned allocation.
We first show that the partition $\bB^+$ computed in Phase $1$ has strong fairness guarantees. 
Recall that in this phase, we only consider items in $M^+$, which induce an instance with identical cost functions.
Consequently, for any $ i, j \in N$, we have $\text{MMS}_i(M^+) = \text{MMS}_j(M^+)$.

\begin{lemma} \label{lemma: Phase1M+EFX}
    For any agent $i \in N$, every bundle $B \in \textbf{B}^+$ is MMS-feasible and EFX-feasible.
\end{lemma}
\begin{proof}
    Recall that we start with an MMS partition of the items $M^+$, and perform reallocations until all bundles in the partition are EFX-feasible.
    Therefore, it suffices to show that the algorithm terminates after a finite number of steps, and the MMS-feasibility of the bundles is maintained throughout the reallocations.

    In the following, we fix an arbitrary agent $i \in N$.
    By the definition of MMS partition, at the beginning of Phase $1$, for any bundle $B\in \bB^+$ we have
    \begin{equation*}
        c_i(B) = c(B) \leq \text{MMS}_i(M^+) \leq \text{MMS}_i(M),
    \end{equation*}
    where the second inequality follows from the fact that $M^+ \subseteq M$. 
    Therefore, to show that all bundles are MMS-feasible to agent $i$ at the end of Phase $1$, it suffices to show that the reallocations do not increase $\max_{j\in [n]}\{ c_i(B_j) \}$.

    Consider any reallocation, in which item $e$ is reallocated from bundle $B_b$ to $B_l$.
    By the condition of the while-loop, we have $c_i(B_b - e) > c_i(B_l)$.
    Thus, we have:
    \begin{equation*}
        c_i(B_l + e) = c(B_l + e) < c(B_b) = c_i(B_b).
    \end{equation*}

    Therefore, after the reallocation, the cost of the two new bundles is upper bounded by $c_i(B_b)$, which implies that $\max_{j\in [n]}\{ c_i(B_j) \}$ does not increase after the reallocation.

    It remains to prove that the while-loops terminate after a finite number of rounds. 
    We show this by introducing a potential function $\sum_{j\in[n]} \left( c(B_j) \right)^2$.
    Clearly, the potential function is upper bounded by $\left( c(M^+) \right)^2$ and lower bounded by $0$.
    Each reallocation of item $e$ from $B_b$ to $B_l$ decreases this potential by
    \begin{align*}
        &\ \left( \hat{c}(B_b) + c(e) \right)^2 + \left( c(B_l) \right)^2 - \left( \hat{c}(B_b) \right)^2 - \left( c(B_l) + c(e) \right)^2 \\
        \geq &\ 2\cdot c(e) \cdot \left( \hat{c}(B_b) - c(B_l) \right) > 0.
    \end{align*}

    Therefore, the potential strictly decreases, implying that Phase $1$ will terminate after a finite number of reallocations.
\end{proof}

\paragraph{Remark.}
Note that we can obtain the same fairness guarantee (for EFX and MMS) if we replace $\bB^+$ by a Leximin allocation of $M^+$\footnote{Here, Leximin refers to the Leximin++ allocation introduced by Plaut and Roughgarden~\cite{journals/siamdm/PlautR20}.}.
However, showing that the Leximin allocation ensures EFX and MMS requires a proof that is almost identical to what we presented above.
Furthermore, while our algorithm is not guaranteed to terminate in polynomial time, it is constructive, which is slightly stronger than the existential guarantee given by Leximin allocations.

\medskip

Next, we analyze Phase $2$. 
In the following two lemmas, we show that if a bundle is modified during this phase, it will be MMS-feasible and EFX-feasible to the agent who performs the last modification, at the end of Phase $2$; if it is not modified by any agent, then it is MMS-feasible and EFX-feasible to all agents.

\begin{lemma} \label{lemma: M0EFX}
    For any bundle $B_j \in \bB$ with $L_j = i \neq 0$, $B_j$ is MMS-feasible and EFX-feasible to agent $i$. 
\end{lemma}
\begin{proof}
    The result follows from the design of Phase $2$.
    In Phase 2, each agent modifies her favorite bundle in turn. 
    By definition, $L_j$ is the last agent who has modified bundle $B_j$.
    Therefore, bundle $B_j$ is agent $i$’s favorite bundle before the modification takes place.
    Furthermore, the modification does not increase the cost of $B_j$ w.r.t. $c_i$. 
    From that moment on, the costs of all other bundles (w.r.t. $c_i$) do not decrease during this phase.
    Consequently, $B_j$ remains agent $i$'s favorite bundle at the end of Phase $2$.
    Therefore, we have
    \begin{equation*}
        c_i(B_j) \leq \max_{k\in [n]} \left\{ c_i(B^+_k) \right\} \leq \text{MMS}_i(M^+) \leq \text{MMS}_i,
    \end{equation*}
    where the first inequality holds since there is at least one bundle that is not modified before agent $i$ performs the modification.
\end{proof}


\begin{lemma} \label{lemma: M+EFX}
    For bundle $B_j \in \bB$ with $L_j = 0$, $B_j$ is MMS-feasible and EFX-feasible to every agent $i \in N$.
\end{lemma}
\begin{proof}
    By Lemma~\ref{lemma: Phase1M+EFX}, every bundle $B_j \in \bB$ with $L_j = 0$ is MMS-feasible and EFX-feasible to every agent $i \in N$ at the beginning of Phase $2$.
    Since $L_j = 0$, bundle $B_j$ is not modified during Phase $2$. 
    Since the cost of each bundle does not decrease during Phase $2$, $B_j$ continues to satisfy the fairness guarantee for every agent $i \in N$ at the end of Phase $2$.
\end{proof}

Finally, we analyze the allocation $\bX$ computed in Phase $3$.
Given Lemma~\ref{lemma: M0EFX} and Lemma~\ref{lemma: M+EFX}, we know that in the final allocation, for all $i \in N$, the bundle $X_i$ assigned to agent $i$ is MMS-feasible and EFX-feasible to agent $i$. 
Therefore, the final computed allocation is MMS and EFX.

\subsection{Polynomial-Time Algorithm}
\label{ssec:poly-time-alg}

In this section, we present a polynomial-time algorithm that computes EFX and $\frac{4}{3}$-MMS allocations (See Algorithm~\ref{alg:load-balancing}).

Note that the main bottleneck in the time complexity of Algorithm~\ref{alg:restricted} lies in Phase $1$\footnote{For identical valuations, the MMS partition always exists. However, computing such a partition is NP-hard.}.
To overcome this, we introduce an alternative method to compute the initial partition in polynomial time, allowing us to replace Phase $1$ while preserving the correctness and structure of Phase $2$ and Phase $3$. 
In particular, Barman and Krishnamurthy~\cite{journals/teco/BarmanK20} and Aziz et al.~\cite{journals/ai/AzizLMWZ24} proved that for IDO instances, the envy-cycle elimination algorithm yields an allocation that is $4/3$-MMS and EFX, respectively. Given that agents have identical valuations in Phase 1, their algorithm becomes equivalent to our Phase 1 procedure, proceeding without any cycle rotation.
\begin{algorithm}[!htb]
\caption{Computation of EFX and $4/3$-MMS Allocations for Restricted Instances} \label{alg:load-balancing}
\KwIn{Restricted instance $I = (N, M, c, \{c_i\}_{i \in N}$)}
$m \gets |M^+|$; \\
Initialize $B_i \gets \emptyset$ for all $i \in [n]$; \\
Sort items in $M^+$ in decreasing order w.r.t. $c$: $\{e_1, \dots, e_m\}$; \\
\For{$j \in [m]$}{
$B_l \gets \arg \min_{B \in \bB} \{c(B)\}$; \\
$B_l \gets B_l + {e_j}$; \\
}
Run Phases $2$ and $3$ on $\bB = (B_1, \dots, B_n)$ to compute allocation $\bX$; \\
\KwOut{EFX and $4/3$-MMS allocation $\bX$}
\end{algorithm}

In the following, we prove the fairness guarantee for the partition $\bB$ computed in Algorithm~\ref{alg:load-balancing}.  
\begin{lemma}
    For any agent $i \in N$, every bundle $B \in \textbf{B}$ is  EFX-feasible and $c_i(B) \leq \frac{4}{3} \cdot \text{MMS}_i(M^+)$.
\end{lemma}
\begin{proof}
We first show that each bundle $B \in \bB$ is EFX-feasible with respect to all agents. 
In each iteration of Algorithm~\ref{alg:load-balancing}, we select the bundle $B_l$ with the least total cost and assign to it the largest remaining item $e$ (since items are sorted in non-increasing order of cost).  
Let $B$ be any other bundle such that $B \neq B_l$. Then we have:
\[
\hat{c}(B_l + e) = c(B_l) \leq c(B),
\]
where the equality holds because $c(e) \leq c(e')$ for all $e' \in B_l$, and the inequality follows from the choice of $B_l$ as the minimum-cost bundle in that iteration.  
This implies that bundle $B_l$ remains EFX-feasible after adding item $e$.
Moreover, for any other bundle $B \neq B_l$, items in $B$ remain unchanged during this step.
Thus, the EFX-feasibility of $B$ is preserved.  
Therefore, all bundles remain EFX-feasible throughout the execution of the algorithm.

Next, we show that $c(B) \leq \frac{4}{3} \cdot \text{MMS}_i(M^+)$.
In fact, our algorithm can be naturally interpreted as an instance of the classical job scheduling problem, where $m$ items correspond to jobs and $n$ bundles correspond to identical machines.
In this setting, our approach is equivalent to applying the Longest Processing Time First (LPT) algorithm, which guarantees a makespan no greater than $\left(\frac{4}{3} - \frac{1}{3n}\right) \cdot \text{OPT}$~\cite{graham1966bounds}.
Therefore, each bundle incurs a cost of at most $\frac{4}{3} \cdot \text{MMS}_i(M^+)$.
\end{proof}

We note that the partition computed in Phase $1$ can be replaced by $\bB$, after which the algorithm proceeds directly to Phases $2$ and $3$. 
The fairness guarantees are preserved throughout Phases $2$ and $3$, following similar arguments to those presented in Section~\ref{Section: Fairness}.
Therefore, the final allocation is EFX and $\frac{4}{3}$-MMS.

We now analyze the overall time complexity of Algorithm~\ref{alg:load-balancing}.
The construction of the initial EFX and $\frac{4}{3}$-MMS partition involves exactly $m$ iterations.
In each iteration, identifying the least costly bundle takes $O(n)$ time, resulting in a total time complexity of $O(nm)$.
In Phase $2$, the for-loop executes exactly $n$ iterations. 
During each iteration, the item set $S_i$ and the bundle $B_j$ can be identified in $O(m)$ and $O(n)$ time, respectively. 
The updates to $B_j$ and $R$ require $O(m)$ time.
In Phase $3$, the final allocation can be computed in $O(n)$.
Therefore, the entire algorithm runs in polynomial time.
Thus, we obtain the following theorem.
\begin{theorem}
    For restricted instances of indivisible chores, there exists an algorithm that computes allocations that satisfy both \emph{EFX} and \emph{$\frac{4}{3}$-MMS} to all agents in polynomial time.
\end{theorem}

Motivated by the incompatibility between EFX and PO, we also study how to maintain fairness guarantees while ensuring PO. 
For general instances, Mahara~\cite{journals/corr/EF1andPO} established the existence of allocations that satisfy both EF1 and PO. 
We show that EF1 and PO allocations can be achieved while simultaneously preserving the MMS guarantee for restricted instances.
\begin{theorem} \label{theorem: EF1andPO}
    For restricted instances of indivisible chores, there exists an algorithm that computes allocations that satisfy both \emph{EF1}, \emph{MMS} and \emph{PO} to all agents.
\end{theorem}

Note that computing an MMS partition for an agent is NP-hard. However, by using a PTAS from \cite{journals/jacm/HochbaumS87}, we can directly obtain the following corollary.
\begin{corollary} \label{corollary: EF1}
    For restricted instances of indivisible chores, there exists an algorithm that computes allocations that satisfy both \emph{EF1}, \emph{$(1-\epsilon)$-MMS}\footnote{$(1-\epsilon)$-MMS allocations guarantee that every agent $i$ receives a bundle whose cost is at most $(1-\epsilon)\cdot \mathrm{MMS}_i$, where $\epsilon>0$.} and \emph{PO} to all agents in polynomial time.
\end{corollary}
Due to the limited space, we defer the proof of Theorem~\ref{theorem: EF1andPO} and Corollary~\ref{corollary: EF1} in Appendix~\ref{section: proofEF1}.

\section{Efficiency Guarantee}
\label{section: Effiency}

In this section, we investigate the efficiency guarantee of Algorithm~\ref{alg:restricted}. 
Since EFX and PO allocations may not exist for restricted instances, we evaluate the social cost of the allocation returned by our algorithm.
Note that the optimal social cost can be achieved by allocating each item to the agent who incurs the least cost for it, for a restricted instance $\mathcal{I}$, the optimal social cost (denoted by OPT($\mathcal{I}$)) is equal to $c(M^+)$.

In the following, we show that our algorithm achieves a $2$-approximation for the optimal social cost. Furthermore, we show that this approximation factor is optimal among algorithms guaranteeing EFX\footnote{The analysis of the efficiency guarantee for Algorithm~\ref{alg:restricted} and Algorithm~\ref{alg:load-balancing} is similar, and both achieve $2$ approximation ratio for the optimal social cost. Therefore, we focus on analyzing the efficiency of Algorithm~\ref{alg:restricted} in the remainder of this section.}.

\begin{theorem} \label{theorem: effiency}
Algorithm~\ref{alg:restricted} achieves a $2$-approximation of the optimal social cost for any restricted instance. Furthermore, the approximation ratio is optimal.
\end{theorem}

We first recall the instance shown in~Table~\ref{Table: NonEFXandPO}.
It is easy to see that the optimal social cost is $1+\epsilon$, while any EFX allocation has a social cost of $2+\epsilon$.
Hence, no EFX allocation can achieve an approximation ratio of optimal social cost better than $2$.

Next, we show that the allocation returned by Algorithm~\ref{alg:restricted} achieves an approximation ratio of $2$ to the optimal social cost, and thus is optimal. 
Let the final allocation be $\textbf{X} = (X_1, \dots, X_n)$, where $X_i$ is assigned to agent $i$.
Recall that during Phase $2$, there could be multiple agents who have modified bundle $X_i$.
Therefore, it is possible that $c_i(X_i)$ is strictly larger than $c_i(X_i \cap M^+)$: this happens when some other agent $j$ modified $X_i$ before agent $i$, during which some items from $M_j^0 \setminus M_i^0$ are added to $X_i$.
Indeed, this is why we can only guarantee an approximation of the optimal social cost, rather than achieving $sc(\bX) = c(M^+)$.

For any bundle $X \in \bX$, we define $d(X)$ as the number of agents who modify $X$ during Phase $2$. 
Based on the values of $d(X_i)$, we partition the agents into three groups:
\begin{itemize}
    \item $N_0 = \{i \in N : d(X_i) = 0\}$;
    \item $N_1 = \{i \in N : d(X_i) = 1\}$;
    \item $N_2 = \{i \in N : d(X_i) \geq 2\}$.
\end{itemize}

Consider the instance illustrated in Figure~\ref{Fig: Algorithm}.
By the design of the algorithm, bundles $B_2$, $B_4$, and $B_6$ are assigned to agents $3$, $4$, and $6$, respectively.
Suppose the remaining bundles $B_1$, $B_3$, and $B_5$ are assigned to agents $1$, $2$, and $5$, respectively.
Then we have $N_0 = \{1, 2, 5\}$, $N_1 = \{4\}$, and $N_2 = \{3, 6\}$.

Note that for all $i\in N_0$, we have $X_i \subseteq M^+$ since it is not modified during Phase $2$; for all $i\in N_1$, we have $X_i \subseteq M^+\cup M^0_i$ since agent $i$ is the only agent who has modified $X_i$.
For these agents we have $c_i(X_i) = c_i(X_i\cap M^+)$.
Therefore, the increase in social cost is due to agents in $N_2$.

Next, we introduce a lemma that upper bounds the number of agents in $N_2$, which will be useful in bounding the total cost.

\begin{lemma} \label{lemma: agentpartition}
    We have $|N_0| \geq |N_2|$.
\end{lemma}
\begin{proof}
Since each of the $n$ agents performs exactly one modification in Phase $2$, the total number of modifications across all bundles in Phase 2 is
\begin{equation*}
    \sum\limits_{i \in N_1} d(X_i) + \sum\limits_{i \in N_2} d(X_i) = n.
\end{equation*}

Since $d(X_i) = 1$ for any agent $i \in N_1$ and $d(X_i) \geq 2$ for any agent $i \in N_2$, we have:
\begin{equation*}
    n = \sum_{i \in N_1} d(X_i) + \sum_{i \in N_2} d(X_i) \geq |N_1| + 2|N_2|,
\end{equation*}
which implies $|N_0| \geq |N_2|$ since $|N_0| + |N_1| + |N_2| = n$.
\end{proof}



Next, we provide upper bounds for the total costs incurred by agents from each group.

\begin{lemma} \label{lemma: BoundofGroupc}
    We have the following properties:
    \begin{itemize}
        \item $\sum\limits_{i \in N_0} c_i(X_i) + \sum\limits_{i \in N_1} c_i(X_i) \leq c(M^+)$;
        \item $\sum\limits_{i \in N_2} c_i(X_i) \leq c(M^+)$.
    \end{itemize}
\end{lemma}
\begin{proof}
    As argued above, for all $i\in N_0\cup N_1$ we have $c_i(X_i) = c_i(X_i\cap M^+) = c(X_i\cap M^+)$, since $X_i$ is either not modified, or only modified by agent $i$.
    Therefore, we can bound the social cost incurred by agents in $N_0$ and $N_1$ by
    \begin{equation*}
        \sum_{i \in N_0} c_i(X_i) + \sum_{i \in N_1} c_i(X_i) 
        = \sum_{i \in N_0 \cup N_1} c(X_i \cap M^+)
        \leq c(M^+).
    \end{equation*}
    
    Next, we show that $\sum_{i \in N_2} c_i(X_i) \leq c(M^+)$. 
    Fix any agent $i \in N_2$.
    By the design of the algorithm, $X_i$ is the most preferred bundle for $i$ in the final allocation, which implies that its cost is no greater than that of any bundle in $\bX$. 
    In particular, we have $c_i(X_i) \leq c_i(X_j)$ for all $j\in N_0$.
    Let $k \in N_2$ be the agent with the maximum cost $c_k(X_k)$. 
    Combined with $|N_0| \geq |N_2|$ by Lemma~\ref{lemma: agentpartition}, we have
    \begin{align*}
    \sum_{i \in N_2} c_i(X_i) 
    &\leq |N_2| \cdot c_k(X_k) \leq  |N_0| \cdot c_k(X_k)\\
    &\leq \sum_{j \in N_0} c_k(X_j) = \sum_{j \in N_0} c(X_j) \leq c(M^+).
    \end{align*}

    where the equality holds since for all $j\in N_0$, $X_j$ contains only items from $M^+$.
\end{proof}

Lemma~\ref{lemma: BoundofGroupc} immediately implies Theorem~\ref{theorem: effiency} since
\begin{equation*}
    sc(\bX) = \sum_{i \in N}c_i(X_i) 
    = \sum_{i \in N_0 \cup N_1}c_i(X_i) + \sum_{i \in N_2}c_i(X_i)
    \leq 2 \cdot c(M^+).
\end{equation*}

Interestingly, our algorithm and analysis imply a tight characterization for the price of fairness (PoF) for EFX.

\begin{definition}[PoF for EFX]
The price of fairness for EFX on instance $\mathcal{I}$ is defined as
\begin{equation*}
    \text{PoF}(\mathcal{I}) = \min\limits_{\text{EFX allocation \textbf{X}}} \ \left\{\frac{sc(\bX)}{\text{OPT}(\mathcal{I})}\right\}, 
\end{equation*}
where $\text{OPT}(\mathcal{I})$ denotes the minimum social cost over all allocations, without any fairness constraint.
Moreover, the price of fairness for EFX, denoted by $PoF$, is defined as the worst-case ratio over all instances, i.e., $\text{PoF} = \sup\limits_{\mathcal{I}} \{\text{PoF}(\mathcal{I})\}$.
\end{definition}

Therefore, we can get the following corollary directly.
\begin{corollary} \label{corollary: PoF}
For restricted instances, the price of EFX is $2$.
\end{corollary}

Furthermore, we also consider the cost of fairness (CoF)~\cite{journals/corr/abs-2410-15738}, which is defined as the worst-case additive difference between the optimal social cost under fairness constraints and the optimal social cost without such constraints, assuming that $c_i(M) \leq 1$ for all $i\in N$.
Equivalently, we define it as follows.

\begin{definition}[CoF for EFX]
    The cost of fairness for EFX is defined as
    \begin{equation*}
        \sup_{\mathcal{I}} \min\limits_{\text{EFX allocation \textbf{X}}} \left\{ \frac{sc(\bX) - \text{OPT}(\mathcal{I})}{\max_{i \in N} c_i(M)} \right\}.
    \end{equation*}
\end{definition}

A hard instance, adapted from Table~\ref{Table: NonEFXandPO}, establishes a lower bound of $1/2$ for the cost of EFX (see Table~\ref{Table: CoF}).
It is easy to see that the optimal social cost is $1/2+\epsilon$, while any EFX allocation has a social cost of $1$.
Hence, the cost of EFX is at least $1 - (1/2+\epsilon) \to 1/2$ when $\epsilon \to 0$.

\begin{table}[h!] 
\centering 
\begin{tabular}{lccc} 
\toprule
\textbf{Agent} & \textbf{$e_1$} & \textbf{$e_2$} & \textbf{$e_3$} \\
\midrule
\ \ \ \ 1  & $1/2 + \epsilon$ & $1/2 - \epsilon$ & $0$ \\
\ \ \ \ 2  & $1/2 + \epsilon$ & $0$ & $1/2 - \epsilon$ \\
\bottomrule
\end{tabular}
\caption{Hard instance for the cost of EFX\\($\epsilon$ is a sufficiently small number)} \label{Table: CoF}
\end{table}

It suffices to show that the upper bound is $1/2$. 
We remark that the properties in Lemma~\ref{lemma: BoundofGroupc} still hold.
To upper bound the cost of EFX, we provide an alternative upper bound for the total cost of agents in $N_2$.

\begin{lemma}
    $\sum\limits_{i \in N_2} c_i(X_i) \leq \frac{1}{2} \cdot \max_{i\in N} c_i(M)$.
\end{lemma}
\begin{proof}
    Note that each agent in $N_2$ receives her favorite bundle.
    For any agent $i\in N_2$, we have 
    \begin{equation*}
        c_i(X_i) \leq \frac{1}{n} \cdot c_i(M) \leq \frac{1}{n} \cdot \max_{j\in N} c_j(M).
    \end{equation*}
    Recall that following Lemma~\ref{lemma: agentpartition} we have $|N_0| \geq |N_2|$, which implies that $|N_2| \leq n/2$.
    Hence we have 
    \begin{equation*}
        \sum\limits_{i \in N_2} c_i(X_i) \leq \frac{1}{2} \cdot \max_{i\in N} c_i(M).  \qedhere
    \end{equation*}
\end{proof}

Consequently, the cost of EFX can be upper bounded by
\begin{align*}
    sc(\bX) - OPT(\mathcal{I}) &= sc(\bX) - c(M^+) \\
    &= \sum_{i \in N_0 \cup N_1}c_i(X_i) + \sum_{i \in N_2}c_i(X_i) - c(M^+) \\
    &\leq \frac{1}{2} \cdot \max_{i\in N} c_i(M).
\end{align*}
This implies that 
$$
\frac{sc(\mathbf{X}) - OPT(\mathcal{I})}{\max_{i \in N} c_i(M)} \leq \frac{1}{2}.
$$


Thus, we obtain the following corollary.
\begin{corollary}\label{corollary:CoF}
For restricted instances, the cost of EFX is ${1}/{2}$.
\end{corollary}
\section{Conclusion and Open Problems}
In this paper, we study the computation of fair and efficient allocations for indivisible chores in the restricted setting. 
We propose an algorithm for computing allocations that are EFX, MMS, and achieve an optimal $2$-approximation of the optimal social cost. 
An important open problem is whether EFX allocations (even without any efficiency guarantee) exist for other classes of instances, e.g., bi-valued setting.
For the allocation of goods, the existence of EFX allocations for restricted instances also remains open.


\newpage
\bibliographystyle{abbrv}
\bibliography{efx}

\newpage
\appendix
\section{Non-existence of EFX allocations for $\infty$-restricted instance} 
In this section, we consider an alternative and natural extension of the restricted setting for chores, which we refer to as the \emph{$\infty$-restricted} instance.
\label{section: infty}
\begin{definition}
    A set $\{c_1, c_2, \dots, c_n\}$ of cost functions is $\infty$-restricted, if for every $1 \leq i \leq n$, for every item $e \in M$, we have $c_i(e) \in \{\infty, c(e)\}$, where $c(e) \geq 0$. Specifically, if $c_i(e) = c(e)$, we say that item $e$ is \emph{relevant} to agent $i$.
\end{definition}
In this model, items cannot be allocated to agents who view them as having infinite cost.
Moreover, we assume that for every item $e \in M$, there exists at least one agent $i$ such that $c_i(e) = c(e)$, ensuring that all items are relevant to at least one agent.

Under these assumptions, we study the existence of EFX allocations subject to the restriction that items with infinite cost cannot be assigned to the corresponding agents.

We consider a special case in this setting, where each item $e$ is considered relevant (with cost $c(e)$) by exactly two agents. Such an instance can be represented as an undirected graph, where each vertex corresponds to an agent and each edge corresponds to an item.
To compute an EFX allocation, we need to find an EFX orientation on the graph, i.e., for each edge, we determine which of its two incident agents receives the corresponding item, such that the resulting allocation is EFX.

We provide an example (see Example~\ref{Example: infty} and Figure~\ref{Figure: infty} for details) that illustrates the nonexistence of EFX allocations in this setting.
\begin{example} \label{Example: infty}
Consider an $\infty$-restricted instance with eight agents and ten items; the structure of agent-item relevance is illustrated in Figure~\ref{Figure: infty}. Blue thick edges represent items with a large cost, e.g., $c(e) = 1000$, while black thin edges represent items with a small cost, e.g., $c(e) = 1$. 
Focusing on agents $1$ through $4$, note that there are six edges between them. By the pigeonhole principle, at least one of these agents must receive more than one item.
Without loss of generality, suppose agent $1$ receives more than one item. Then, regardless of how we allocate the item shared between agents $1$ and $5$, agent $1$ will not be EFX towards agent $5$, due to receiving multiple large-cost items.
\end{example}

\begin{figure}[h]
\centering
\begin{tikzpicture}[scale=1.2]

\tikzstyle{mynode} = [circle, draw, minimum size=14pt, inner sep=1pt, font=\small]

\node[mynode] (v1) at (-0.7, 0.8) {1};
\node[mynode] (v2) at (0.7, 0.8) {2};
\node[mynode] (v3) at (0.7, -0.8) {3};
\node[mynode] (v4) at (-0.7, -0.8) {4};

\foreach \i/\j in {1/2, 1/3, 1/4, 2/3, 2/4, 3/4} {
    \draw[blue, thick] (v\i) -- (v\j);
}

\node[mynode] (u1) at (-1.5, 0.8) {5};
\node[mynode] (u2) at (1.5, 0.8) {6};
\node[mynode] (u3) at (1.5, -0.8) {7};
\node[mynode] (u4) at (-1.5, -0.8) {8};

\foreach \i in {1,...,4} {
    \draw (v\i) -- (u\i);
}

\end{tikzpicture}
\caption{Illustration of Example~\ref{Example: infty}. }\label{Figure: infty} 
\end{figure}
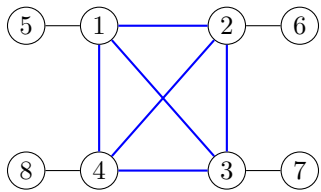

\section{EF1 + MMS + PO Allocations} 
\label{section: proofEF1}

\begin{algorithm}[!htb]
\caption{Computation of EF1, MMS, and PO Allocations for Restricted Instances}
\label{alg: EF1}
\KwIn{Restricted instance $I = (N, M, c, \{c_i\}_{i \in N}$)}
    Initialize $B_i \gets \emptyset$ for all $i \in [n]$; \\
    \tcp{Phase 1: Allocation of items in $M^+$} 
Compute an MMS partition $\textbf{B} = (B_1, \dots, B_n)$ of $M^+$ with respect to $c$; \\[4pt]

\tcp{Modify the partition to ensure EF1}
\While{$\max_{B \in \bB} \bar{c}(B) > \min_{B \in \bB} c(B)$ \label{While-loop: EF1}}{ 
    $B_b \gets \arg\max_{B \in \bB} \bar{c}(B)$; \\
    $B_l \gets \arg\min_{B \in \bB} c(B)$; \\
    $e \gets \arg\max_{e' \in B_b} c(e')$; \\
    $B_l \gets B_l + e$;\\    
    $B_b \gets B_b - e$;\\
}

\tcp{Phase 2: Allocation of items in $M^0$}
$R \gets M^0$; \\
\For{$i \in N$}{
    $S_i \gets M_i^0 \cap R$; \\
    Pick an arbitrary bundle $B \in \textbf{B}$; \\
    $B \gets B \cup S_i$; \\
    $X_i \gets B$; \\
    $R \gets R \setminus S_i$; \\
    $\textbf{B} \gets \textbf{B} \setminus \{B\}$; \\
}
\KwOut{$\bX$ that satisfies EF1, MMS, and PO simultanously}
\end{algorithm}
To prove Theorem~\ref{theorem: EF1andPO}, we present Algorithm~\ref{alg: EF1} and introduce several key observations.
For convenience, for agent $i$ and bundle $B$, we further define
\begin{equation*}
    \bar{c}_i(B) = \min_{e \in B}\{c_i(B - e)\}
\end{equation*}
as the bundle's cost after removing the item with the largest cost.

We refer to one iteration of the while-loop as a round, and index the rounds by $t = 1, 2, \dots$. 
We use $B_b^t$ and $B_l^t$ to denote the bundles that we identify at the beginning of round $t$.
We refer to $B_b^t$ and $B_l^t$ as the big bundle and the least bundle in round $t$, respectively\footnote{When selecting the big bundle, we break ties by choosing the bundle with the smallest index. Therefore, we refer to it as "the" big bundle, with a similar rule applied to "the" least bundle.}.
\begin{observation} \label{observation: nonincreasing}
Let the while-loop on Line~\ref{While-loop: EF1} run for a total of $t$ rounds. Then for all $1 \leq t' \leq t - 1$, we have $\bar{c}(B_b^{t'}) \geq \bar{c}(B_b^{t'+1})$ and $c(B_l^{t'}) \leq c(B_l^{t'+1})$.
\end{observation}
\begin{proof}
    In round $t'$, we select bundles $B_b^{t'}$ and $B_l^{t'}$, and transfer the maximum-cost item $e$ (w.r.t. $c$) from $B_b^{t'}$ to $B_l^{t'}$.  
    As a result, the cost of $B_b^{t'}$ does not increase in this round.
    Moreover, since we have $c(B_l^{t'} \cup \{e\}) < c(B_b^{t'})$ by the while-loop condition, the updated cost of $B_l^{t'}$ does not exceed the previous cost of $B_b^{t'}$.
    The items of all other bundles remain unchanged in this round. 
    Therefore, the value of $\max_{i \in [n]}\{\bar{c}(B_i^{t'})\}$ does not increase in this round, which implies that $\bar{c}(B_b^{t'}) \geq \bar{c}(B_b^{t' + 1})$.
    On the other hand, the updated cost of $B_b^{t'}$ does not lower than the previous cost of $B_l^{t'}$, which implies that the minimum cost (w.r.t. $c$) will not decrease in this round, leading to the fact that $c(B_l^{t'}) \leq c(B_l^{t' + 1})$.
\end{proof}

\begin{observation} \label{Observation: polynomialtime}
    If a bundle $B_b^t$ is selected as the big bundle for the last time in round $t$, then it can never be selected as the least bundle in any subsequent round.
\end{observation}
\begin{proof}
    Suppose for contradiction that $B_b^t$ is selected as the least bundle in some later round $t' > t$, and let $t'$ be the smallest such round.
    Assume that in round $t'$, the bundle $B_b^{t'}$ is selected as the big bundle, and $B_l^{t'} = B_b^t$ is selected as the least bundle.
    Since round $t$ is the last round that $B_b^t$ is selected as the big bundle, we have $c(B_l^{t'}) = \bar{c}(B_b^t)$, i.e., items in bundle $B_b^t$ remain unchanged from round $t$ to $t' - 1$.
    Then in round $t'$, we have:
    $$
    \bar{c}(B_b^{t'}) > c(B_l^{t'}) = \bar{c}(B_b^t).
    $$
    However, by Observation~\ref{observation: nonincreasing} and the fact that $t < t'$, we know that
    $$
        \bar{c}(B_b^t) \geq \bar{c}(B_b^{t'}),
    $$
    which leads to a contradiction.

    Hence, a bundle that is selected as the big bundle for the last time in some round can never be selected as the least bundle in any subsequent round.
\end{proof}

\begin{lemma} \label{lemma: whileloop}
    The while-loop on Line~\ref{While-loop: EF1} runs at most $nm$ rounds.
\end{lemma}
\begin{proof}
In each round, a bundle is selected as the big bundle, and exactly one item is removed from it.
Hence, the size of the bundle decreases by exactly one each time it is selected as the big bundle.
By Observation~\ref{Observation: polynomialtime}, once a bundle is selected as the big bundle, it can never be selected as the least bundle in any subsequent round, which implies that any bundle can be selected as the big bundle at most $m$ times.
Since there are $n$ bundles, the while-loop on Line~\ref{While-loop: EF1} terminates after at most $nm$ rounds.

We are now ready to prove Theorem~\ref{theorem: EF1andPO}.

By conducting an analysis similar to the one in Lemma~\ref{lemma: M0EFX}, we can show that at the end of Phase 1, every bundle is MMS-feasible. 
Additionally, as established in Lemma~\ref{lemma: whileloop}, the while-loop will eventually terminate. 
The condition of the while-loop ensures that every bundle also satisfies EF1.

In Phase 2, each agent selects a bundle and adds items that impose zero cost for them, which guarantees that both MMS and EF1 are preserved. 
Furthermore, since the final allocation minimizes the social cost, it guarantees Pareto Optimality. 
Therefore, the final allocation computed by Algorithm~\ref{alg: EF1} satisfies EF1, MMS, and PO simultaneously.

Finally, if we use a PTAS to compute the ($1-\epsilon$)-MMS partition in Phase 1 instead, we can guarantee that the overall algorithm runs in polynomial time, thereby proving Corollary~\ref{corollary: EF1}.
\end{proof}

\end{document}